\documentclass[a4paper,12pt]{article}
\usepackage{amsmath}
\usepackage{graphicx}
\usepackage{rotating}
\usepackage{wasysym}
\DeclareGraphicsExtensions{.pdf,.png,.jpg}
\usepackage{ chicago}
\usepackage{ klett}
\usepackage{ def_the}

\usepackage{ Fig_size_A}
\usepackage{ kfig}

\newtheorem{corollary}[theorem]{Corollary}
\newtheorem{assumption}[theorem]{Assumption}

\title{}
\author{Eckhard Platen}
\begin{document}
\thispagestyle{empty} \vspace*{1.0cm}

\begin{center}
{\LARGE\bf  Benchmark-Neutral Pricing}
\end{center}

\vspace*{.5cm}
\begin{center}

{\large \renewcommand{\thefootnote}{\arabic{footnote}} {\bf Eckhard
Platen}\footnote{University of Technology Sydney,
  School of Mathematical and Physical Sciences}}
\vspace*{2.5cm}

\today

\end{center}

\begin{minipage}[t]{13cm}
The paper introduces benchmark-neutral pricing and hedging for long-term   contingent claims.  It employs the growth optimal portfolio  of the stocks   as num\'eraire and the  new benchmark-neutral  pricing measure for pricing. For a realistic  parsimonious  model, this  pricing measure turns out to be an equivalent probability measure, which is not the case for the risk-neutral pricing measure. Many risk-neutral prices of long-term contracts are more expensive than necessary.   Benchmark-neutral pricing identifies  the minimal possible prices  of contingent claims, which is  illustrated with remarkable accuracy for a long-term zero-coupon bond.  
\end{minipage}
\vspace*{0.5cm}


\vspace*{0.5cm}
\vspace*{0.5cm}

	\newpage
	\section{Introduction}\label{section.intro}
Risk-neutral pricing employs the savings account as num\'eraire and represents the preferred pricing method of the classical  finance theory, see, e.g. , \citeN{Merton92}, \citeN{Cochrane01}, and \citeN{Jarrow22}. 
 By employing  realistic long-term dynamics of the {\em growth optimal portfolio} (GOP) of stocks,  see, e.g., \citeN{Kelly56} and \citeN{Merton92}, the current paper  shows   that the putative risk-neutral measure is unlikely to constitute an equivalent probability measure. These findings are not in line  with 
  the {\em No Free Lunch with Vanishing Risk} (NFLVR) no-arbitrage condition of \citeN{DelbaenSc98}, which is equivalent to the assumption about the existence of an equivalent risk-neutral probability measure. In \citeN{LoewensteinWi00a} and \citeN{Platen01a} it was pointed out that perfectly acceptable financial market models exist where the NFLVR condition fails. Under the more general benchmark approach, see \citeN{Platen06ba} and \citeN{PlatenHe06}, which removes a restrictive assumption of the classical finance theory, the failure of the NFLVR condition  does not represent a problem. The benchmark approach only  requires  the existence of the GOP and not the existence of an equivalent risk-neutral probability measure.	This matters because when the putative risk-neutral measure is not an equivalent probability measure, it is shown in \citeN{Platen02g} and \citeN{Platen06ba}  that risk-neutral pricing leads to more expensive prices and hedges than necessary.  \\ The existence of the GOP  is an extremely weak and easily verifiable {\em no-arbitrage condition}   because \citeN{KaratzasKa07} and \citeN{KaratzasKa21} have shown that  the existence of the GOP is equivalent to their {\em No Unbounded Profit with Bounded Risk} (NUPBR)  condition. This no-arbitrage condition is weaker than the NFLVR  condition. The current paper will employ a realistic long-term model of the {\em stock GOP}, which is the GOP of the investment universe formed only by the stocks  (without the savings account). For this realistic model the NUPBR condition holds but the NFLVR condition fails.\\

 \noindent The benchmark approach provides  the pricing concept of real-world pricing, where the GOP of the entire market is taken as the num\'eraire and the real-world probability measure acts as the pricing measure; see \citeN{PlatenHe06}. Real-world pricing avoids the additional assumptions that a change of the GOP of the entire market as num\'eraire would require. 
 \\ For a typical market that consists of stocks and the savings account, the GOP of the entire market is a highly leveraged portfolio that goes short in the savings account and long in a portfolio of stocks; see \citeN{FilipovicPl09}. Unfortunately, this GOP  does  not represent a suitable num\'eraire or even a desireable investment portfolio; see e.g., \citeN{Samuelson79}. To generate it physically one must go  short in the savings account. Since only  discrete-time dynamic asset allocation is feasible, one could obtain    as a proxy of the GOP of the entire market  a negative portfolio, which would fail as a num\'eraire for pricing and hedging. \\
 As a  feasible alternative to real-world pricing, which provides the minimal possible prices, the current paper proposes  employing the  stock GOP as a num\'eraire. As shown in \citeN{PlatenRe20}, the  stock GOP can be approximated by a well-diversified, guaranteed strictly positive portfolio of stocks. For instance, the MSCI-Total Return Stock Index (MSCI)  of the developed markets could serve as a reasonable   proxy of the stock GOP. However, better proxies are available, as shown in \citeN{PlatenRe20}. \\
 	The proposed new pricing  method is called {\em benchmark-neutral} (BN) {\em pricing}. It has wider applicability than risk-neutral pricing and can be conveniently implemented, as will be demonstrated in the current paper. It provides the minimal possible prices  of nonnegative contingent claims when  
 	the ratio of the stock GOP over the GOP of the entire market is a martingale. This ratio represents the Radon-Nikodym derivative of the respective {\em BN pricing measure} and is theoretically  a supermartingale. This means, its current value is greater than or equal to its expected future values given the current information.
When appropriate dynamics of the stock GOP are assumed, as will be suggested in the current paper, the above-mentioned ratio turns out to be a martingale, and the BN prices of contingent claims coincide with the respective minimal possible prices obtained via the real-world pricing formula; see  \citeN{DuPl16}.\\
	 Motivated by the structure of the stochastic differential equation (SDE) of the stock GOP of a continuous market,  the {\em minimal market model} (MMM) was proposed in \citeN{Platen01a} as a potential model for the stock GOP.   The current paper points  out that the stock GOP value under the MMM can be interpreted as the continuous time limit of the population size of  a birth-and-death process, see \citeN{Feller71}, where independently  wealth units give  birth to new ones or die.\\ As the  paper will demonstrate, the MMM evolving in some {\em activity time}  captures remarkably  well the \textquoteleft natural' evolution of well-diversified stock indexes. 	It models parsimoniously the volatility
	  of the  normalized stock GOP as a scalar  diffusion process that is evolving in some activity time. The latter is reflecting the trading activity and can be observed but only its linear average needs to be modeled for the pricing and hedging of long-term zero-coupon bonds. This circumvents  the need to observe and model all parts of the volatility, which  seems to be extremely difficult  as pointed out by the \textquoteleft leverage effect puzzle' of \citeN{AitSahalia12}. \\
	  The current paper points in a direction where this puzzle, and more general stochastic volatility modeling, could find  solutions. It 
	  applies the MMM in some activity time 	 	
	 	for the stock GOP, and illustrates the fitting of its parameters, as well as    the BN pricing and accurate hedging of a long-term zero-coupon bond that pays one unit of the savings account at maturity. \\
	 	 In previous works,
	 	 long-term zero-coupon bonds and other  derivatives have been priced and hedged using the MMM when it evolves in calendar time and by applying real-world pricing  assuming that the GOP of the entire market equals the stock GOP; see, e.g., \citeN{PlatenHe06}, \citeN{FergussonPl22} and \citeN{BaroneadesiPlSa24}. The novelty of the current paper is that the GOP of the entire market is allowed to be  different from the stock GOP. Furthermore, the stock GOP  dynamics are far more realistically modeled by introducing some  flexible  stochastic activity time for the MMM. The activity time is observable and only its trendline needs to be estimated for the pricing and hedging of long-term zero-coupon bonds. This is  important for  implementations of BN pricing and hedging in the financial services industry because long-term zero coupon bonds can serve as the building blocks of long-term annuities, pensions,  green bonds, and life insurance contracts. \\

	     The paper is organized as follows:  Section 2  introduces  the market setting and real-world pricing under the benchmark approach.
	      The new concept of benchmark-neutral pricing is presented in Section 3. 
	       The  model of the stock GOP is introduced in Section 4. Section 5   illustrates the BN pricing and   hedging of a long-term zero-coupon bond.
	       
	\section{Market Setting}\setcA
\subsection{Primary Security Accounts}
	 The  modeling is performed on a filtered probability space $(\Omega,\mathcal{F},\underline{\cal{F}},P)$, satisfying the usual conditions; see, e.g., \citeN{KaratzasSh88} and \citeN{KaratzasSh98}. The filtration $\underline{\cal{F}}=(\mathcal{F}_t)_{t\in[t_0,\infty)}$ describes the evolution of market information
	 over time. In the given continuous  market we model 
	 $d\in \{1,2,...\}$ adapted, nonnegative {\em primary security accounts}, denoted by $S^1_t,S^2_t,...,S^d_t$, 
	 where all  dividends or other payments are reinvested. We interpret the $d$ primary security accounts as stocks, which are above denominated   in units of the savings account $S^0_t=1$.  We assume for the investment universe consisting of the $d$ stocks  that a   {\em growth optimal portfolio} (GOP) $S^*_t$ exists, which   is the strictly positive stock portfolio with maximum growth rate;  see \citeN{Kelly56} and \citeN{Merton92}. This portfolio we call the {\em stock GOP}. The $j$-th  primary security account ${\tilde S}^j_t=\frac{S^j_t}{S^*_t}$, $j\in\{1,...,d\}$, when denominated in units of the stock GOP $S^*_t$,   forms a continuous, integrable driftless stochastic process, which is a $(P,\underline{\cal{F}})$-local
 martingale and, thus, a  $(P,\underline{\cal{F}})$-supermartingale; see equation (10.3.2) in \citeN{PlatenHe06}. It has to be emphasized that it has not to be a martingale where its current value would equal its expected future values under the current information.
	    We emphasize that  the savings account $S^0_t$ is not included in the investment universe of the stock GOP $S^*_t$.\\ 
	   The stock GOP $S^*_t$ in savings account denomination is continuous and satisfies according to
	      Theorem 3.1 in \citeN{FilipovicPl09} the stochastic differential equation (SDE)
	      \begin{equation}\label{dS*T}
	\frac{dS^{*}_t}{S^{*}_t}=\lambda^*_t dt +
	 \theta_t  (  \theta_t dt+d W_t) \end{equation}
	 for $t\in[t_0,\infty)$ with $S^*_{t_0}>0$.
Here, $\lambda^*_t$ denotes the {\em net risk-adjusted return} of  $S^*_t$  and $\theta_t$ its {\em volatility}. The real-valued $(P, \mathcal{\underline F})$-Brownian motion $W=\{W_t, t\in[t_0,\infty)\}$ 
   models  in calendar time the non-diversifiable randomness of the stocks in savings account denomination. 
\\
We  extend the above market of stocks by adding the savings account $S^0_t=1$ as an additional primary security account and assume that the GOP     $S^{**}_t$  of the extended market exists.  In line with Theorem 7.1 in \citeN{FilipovicPl09}, we assume  that the GOP  $S^{**}_t$ of the extended market  satisfies the SDE \BE \label{e.4.2'}
\frac{dS^{**}_t}{S^{**}_t}=
\sigma^{**}_t(  \sigma^{**}_t dt+d W_t) \EE
 with initial value $S^{**}_{t_0}=1$ and {\em market price of risk}
\begin{equation}\label{sigma**}
\sigma^{**}_t=\frac{\lambda^*_t  +
	(\theta_t)^2  }{\theta_t}
\end{equation}  with respect to $W$ for $t\in[t_0,\infty)$. 
\subsection{Real-World Pricing}
	Consider a bounded stopping time $T > t_0$, and let $\mathcal{L}^1(\mathcal{F}_T)$
	denote the set of integrable, $\mathcal{F}_T$-measurable random variables in the given filtered probability space.
	\begin{definition}\label{def2.1}
		 For a bounded stopping time $T \in(t_0,\infty)$, a 
		  nonnegative payoff $H_T$, denominated in units of the savings account, is called a {\em contingent claim} if $\frac{H_T}{S^{**}_T} \in
	 \mathcal{L}^1(\mathcal{F}_T)$.
	\end{definition}
	 	We denote by ${\bf E}^{P}(.|\mathcal{F}_t)$  the conditional expectation under the real-world probability measure $P$, conditional on the information $\mathcal{F}_t$ available at time $t$.  As shown in \citeN{PlatenHe06} and \citeN{DuPl16},
	  for a
		 contingent claim $\ H_T$ with maturity at a bounded stopping time $T$ the {\em real-world pricing formula}
		  \begin{equation}
		   H_t=S^{**}_t{\bf E}^{P}(\frac{{H}_{T}}{S^{**}_T}|\mathcal{F}_t)\end{equation} determines its, so called, {\em fair}  price $ H_t$ for all $t\in[t_0,T]$. The ratio $\frac{H_t}{S^{**}_t}$ forms a $(P, \mathcal{\underline F})$-martingale, which is $P$-almost surely unique for the given value $\frac{H_T}{S^{**}_T}$ at  the maturity  $T$.	The real-world pricing formula uses the GOP $S^{**}_t$ of the extended market  as num\'eraire and the real-world probability measure $P$ as pricing measure.  It has been shown by \citeN{DuPl16} that, when a contingent claim is replicable,  its fair price process coincides with the value process of the minimal possible self-financing hedge portfolio that replicates its payoff.\\
		  The Law of One Price of the classical finance theory does no longer hold because there exist  other  pricing rules that can be applied  to pricing and hedging, including the popular risk-neutral pricing rule. 
		  However, these  pricing rules  never provide  lower prices for nonnegative replicable contingent claims than the real-world pricing formula because all self-financing portfolios that hedge a given contingent claim form $(P, \mathcal{\underline F})$-supermartingales when denominated in $S^{**}_t$. The $(P, \mathcal{\underline F})$-martingale among these $(P, \mathcal{\underline F})$-supermartingales coincides with the  fair hedge portfolio value process in the denomination of the GOP $S^{**}$ of the extended market. It  is the least expensive hedge portfolio that replicates the contingent claim; see Lemma A.1 in \citeN{DuPl16}. More expensive self-financing hedge portfolios can exist that replicate the contingent claim. Since these  replicate the contingent claim, one may not even realize that they are more expensive than necessary.

	 \section{Benchmark-Neutral Pricing}\setcA
	 
	 \subsection{Change of Num\'eraire} 
	  The num\'eraire for real-world pricing is the GOP $S^{**}_t$ of the extended market, which is, in reality, a highly leveraged portfolio that goes  long in the stock GOP $S^*_t$ and  short in the savings account $S^0_t$. When  hedging  contingent claims, one needs to be able to trade  the num\'eraire that one is using for pricing and hedging. For instance, when hedging a zero-coupon bond that pays one unit of the savings account at maturity, the hedging requires the trading of the num\'eraire and the savings account, as  will be shown later. 
	  \\ A tradeable proxy of the highly leveraged GOP $S^{**}_t$ of the extended market cannot be easily constructed as a  guaranteed strictly positive, self-financing portfolio because such a highly leveraged portfolio  can only be traded at discrete  times and, therefore,  faces  the possibility of becoming negative. \\  
	 To avoid the above difficulties, the paper suggests  employing the 
	  stock GOP $S^*_t$ as a num\'eraire.   As shown in \citeN{PlatenRe20}, a well-diversified total return stock index can be a reasonable  proxy for the stock GOP and is, by construction,  strictly positive.   Total return stock indexes have been  used traditionally as benchmarks in fund management. The current paper suggests  employing such a benchmark  as a num\'eraire for pricing and hedging. It calls the new pricing method  {\em benchmark-neutral  pricing} (BN pricing) and the proxy for the stock GOP the {\em benchmark}. Intuitively, in the denomination of the benchmark and under the respective pricing measure the expected returns of portfolios are zero and, in this sense, \textquoteleft neutral' to the  randomness that  drives the market.\\ The stock GOP has, in the long run, a trajectory that is almost surely pathwise outperforming any other strictly positive stock portfolio; see Theorem 10.5.1 in \citeN{PlatenHe06}. Intuitively, BN pricing  centers the risk management   around the  long-run best-performing strictly positive stock portfolio, whereas risk-neutral pricing centers it around the rather poorly performing savings account.\\
	  
 \noindent By application of the It\^{o} formula it follows that the stock GOP $S^*_t$, when denominated  in units of the GOP $S^{**}_t$, satisfies a driftless SDE and is, therefore, a  $(P, \mathcal{\underline F})$-local martingale. We make throughout the paper the following assumption, which we verify later  for  the realistic stock GOP  model that we will employ:
 \begin{assumption}\label{RNDMARTINGALE}
 	The stock GOP $S^*_t$, when denominated  in units of the GOP $S^{**}_t$, forms the  $(P, \mathcal{\underline F})$-martingale $\frac{S^*}{S^{**}}=\{\frac{ S^*_t}{S^{**}_t},t\in[t_0,\infty)\}$.
 \end{assumption} 	Under this assumption, the proposed change of num\'eraire permits the application of the Change of Num\'eraire Theorem by \citeN{GemanElRo95}. This leads to the  Radon-Nikodym derivative \begin{equation}\label{RND}
	 \Lambda_{S^*}(t)=\frac{dQ_{S^*}}{dP}|_{\mathcal{F}_t}=\frac{\frac{ S^*_t}{S^{**}_t}}{\frac{ S^*_{t_0}}{S^{**}_{t_0}}},\end{equation} which characterizes for the num\'eraire $S^*_t$ the respective {\em benchmark-neutral pricing measure} $Q_{S^*}$. By application of the It\^{o} formula we obtain from \eqref{dS*T} and \eqref{e.4.2'} the SDE
	 \begin{equation}
	 \frac{d\Lambda_{S^*}(t)}{\Lambda_{S^*}(t)}
	 =-\sigma^{S^*}(t)dW_t
	 \end{equation}
	 	with
	 	\begin{equation}\label{sigma10}
	 	\sigma^{S^*}(t)=\frac{\lambda^*_t}{\theta_t}
	 	\end{equation}
	 	for $t\in[t_0,\infty)$.  As a self-financing portfolio that is denominated in units of the GOP $S^{**}_t$ of the extended market, the Radon-Nikodym derivative $\Lambda_{S^*}(t)$ forms in any case a $(P, \mathcal{\underline F})$-local martingale. Under Assumption \ref{RNDMARTINGALE}  it is assumed to be a true martingale.  Let  ${\bf E}^{Q_{S^*}}(.|\mathcal{F}_t)$ denote  the conditional expectation  with respect to the benchmark-neutral pricing measure $Q_{S^*}$ under the information available at time $t\in[t_0,\infty)$. We obtain  directly from \citeN{GemanElRo95} and Girsanov's Theorem, given in, e.g., \citeN{KaratzasSh88} and \citeN{KaratzasSh98}, the following result:

\begin{theorem}[Benchmark-Neutral Pricing Formula]
	 Under the Assumption \ref{RNDMARTINGALE}, the benchmark-neutral pricing measure $Q_{S^*}$ is an equivalent probability measure, and the fair price $H_t$, which the real-world pricing formula identifies for a  contingent claim  $ H_T$, is  obtained via the {\em benchmark-neutral pricing formula} 
		\BE \label{BPF}H_t =S^*_t {\bf E}^{Q_{S^*}}(\frac{H_{T}}{S^*_T}|\mathcal{F}_t) \EE
	 for $t\in[t_0,T]$.   The process $W^{0}=\{W^{0}_t, t\in[t_0,\infty)\}$,
	 satisfying the SDE
	 \begin{equation}\label{W0}
	 dW^{0}_t=\sigma^{S^*}(t)dt+dW_t
	 \end{equation}
 for $t\in[t_0,\infty)$ with $W^{0}_{t_0}=0$, is under  $Q_{S^*}$ a Brownian motion with respect to calendar time, and the  primary security accounts $\tilde S^0_t,...,\tilde S^d_t$,  denominated in the stock GOP $S^*_t$, represent $(Q_{S^*}, \mathcal{\underline F})$-local martingales.
\end{theorem}
This result is of  practical importance because it allows one to use the stock GOP  as a num\'eraire for  pricing and hedging under the benchmark-neutral pricing measure $Q_{S^*}$. Since  the process $W^{0}$ can be interpreted as a Brownian motion under this measure,  the respective SDE for the stock GOP $S^*_t$ takes by \eqref{dS*T}, \eqref{sigma10}, and \eqref{W0} the form
	\BE \label{e.4}
	dS^{*}_t=S^{*}_t
	\theta_t  (  \theta_t dt+d W^{0}_t) \EE
	for $t\in[t_0,\infty)$ with $S^*_{t_0}>0$. The stock GOP-denominated savings account $\tilde S^0_t=\frac{S^0_t}{S^*_t}=\frac{1}{S^*_t}$ satisfies, by application of the It\^{o} formula, the SDE
		\BE \label{e.4''}
		d\tilde S^{0}_t=-\tilde S^{0}_t
		\theta_t  d W^{0}_t \EE
		for $t\in[t_0,\infty)$ with $\tilde S^0_{t_0}>0$, which confirms that $\tilde S^0_t$ forms a $(Q_{S^*}, \mathcal{\underline F})$-local martingale.\\
		 BN pricing is capturing the dynamics under $Q_{S^*}$  as if under $P$ the net risk-adjusted return $\lambda^*_t$ were zero. Consequently, there is no need  to model or estimate the risk-adjusted return, which simplifies significantly the modeling and pricing.  Only the  real-world Brownian motion $W$, which models the nondiversifiable randomness, attracts in the denomination of the savings account a risk premium. The other randomness in the market evolves  under $Q_{S^*}$ as it does under the real-world probability measure $P$. For this reason we focus below in our illustration on a long-term zero-coupon bond that has as contingent claim  a function of the stock GOP. 
\subsection{Portfolios}
The market participants can combine primary security accounts to form portfolios.
Denote by $\delta = \{\delta_t = (\delta^0_t,\delta^1_t,...,\delta^d_t)^\top, t\in[t_0,\infty)\}$
 the strategy, where $\delta^j_t$, $j\in\{0,1,...,d\}$, represents the number of units of the $j$-th primary security account that are held at time $t\in[t_0,\infty)$ in a
corresponding portfolio $S^\delta_t$. When denominated in units of the stock GOP $S^{*}_t$, this  portfolio is captured
by the   process $\tilde S^\delta = \{\tilde S^\delta_t=\frac{S^\delta_t}{S^*_t},t\in[t_0,\infty)\}$,
 where \begin{equation}
 \tilde S^\delta_t=(\delta_t)^\top{\bf{\tilde S}}_t
 \end{equation}
for $t\in[t_0,\infty)$ with  ${\bf{\tilde S}}_t=(\tilde S^0_t,...,\tilde S^d_t)^\top$.
If changes in the value of a portfolio are only due to changes in the values of the
primary security accounts, then no extra funds flow in or out of the portfolio, and
the corresponding portfolio and strategy are called {\em self-financing}. The self-financing property of a  portfolio is expressed by the equation
\begin{equation}\label{stochint}
\tilde S^\delta_t=\tilde S^\delta_{t_0}+\int_{t_0}^{t}(\delta_s)^\top d{\bf{\tilde S}}_s
\end{equation}
for $t\in[t_0,\infty)$ with $ \tilde S^\delta_{t_0}=(\delta_{t_0})^\top{\bf{\tilde S}}_{t_0}$, where the stochastic integral in \eqref{stochint} is assumed to be a vector-It\^{o} integral; see \citeN{ShiryaevCh02}. \\
 To
 introduce a class of admissible strategies for forming
  portfolios, denote by $[{\bf{\tilde S_.}}]_t=([\tilde S^i_.,\tilde S^j_.]_t)_{i,j=0}^{d} $ the matrix-valued optional covariation  of the vector  of stock GOP-denominated primary security
accounts ${\bf{\tilde S}}_t$ for  $t\in[t_0,\infty)$.

 \begin{definition}\label{integr'}
 	An {\em admissible self-financing strategy} ${\bf {\delta}}=\{\delta_t=(\delta^0_t,...,\delta^d_t)^\top,\\t\in[t_0,\infty)\}$, initiated at the time $ t_0$, is an ${\bf{R^{d+1}}}$-valued, predictable
 	stochastic process, satisfying the condition
 	\begin{equation}\label{integrability'}
 	\int_{t_0}^{t}\delta^\top_u[{\bf{\tilde S_.}}]_u \delta_u du<\infty
 	\end{equation} 
 for  $t\in[t_0,\infty)$.
 	\end{definition}
An admissible self-financing strategy generates 
the stock GOP-denominated gains from trade
\begin{equation}\label{stochint'}
\int_{t_0}^{t}\delta_s^\top d{\bf{\tilde S_s}}=\int_{t_0}^{t}d \tilde S^\delta_s=\tilde S^\delta_t-\tilde S^\delta_{t_0}
\end{equation}
for $t\in[t_0,\infty)$. 
It does this without requiring outside funds or generating
extra funds. The predictability of the integrand in the above stock GOP-denominated gains from trade  expresses
the real informational constraint that the allocation of units of primary security accounts in the admissible self-financing strategy $\delta$ is not allowed to
anticipate the movements of the stock GOP-denominated primary security account vector ${\bf{\tilde S_t}}$.

\subsection{Contingent Claims}
In the following, we consider contingent claims that can be  replicated by using self-financing portfolios under BN pricing. Let for a bounded stopping time $T$ the set $ \mathcal{L}^1_{Q_{S^*}}(\mathcal{F}_T)$ denote the set of $\mathcal{F}_T$-measurable and $Q_{S^*}$-integrable random variables in the filtered probability space  $(\Omega,\mathcal{F},\underline{\cal{F}},Q_{S^*})$.\\

\begin{definition}\label{HT}
	We call for a bounded stopping time $T\in[t_0,\infty)$ a stock GOP-denominated contingent claim  $\tilde H_T=\frac{H_T}{S^*_T}\in \mathcal{L}^1_{Q_{S^*}}(\mathcal{F}_T)$ {\em BN-replicable}   if it
	has for all $t\in[t_0,T]$ a representation of the  form
	\begin{equation}\label{tildeHT}
	\tilde H_T={\bf E}^{Q_{S^*}}(\tilde H_{T}|\mathcal{F}_t)+\int_{t}^{T}\delta_{\tilde H_T}(s)^\top d{\bf{\tilde S_s}}
	\end{equation}
$Q_{S^*}$-almost surely, involving some predictable vector process $\delta_{\tilde H_T}=\{\delta_{\tilde H_T}(t)=(\delta_{\tilde H_T}^0(t),...,
\delta_{\tilde H_T}^d(t))^\top,t\in[t_0,T]\}$  satisfying the condition \eqref{integrability'}.
\end{definition}
To capture the replication of a targeted contingent claim we introduce the following notion:
\begin{definition} \label{deliver}
	We say, an admissible self-financing  strategy $\delta=\{\delta_t=(\delta^0_t,...,\\ \delta_t^d)^\top,t\in[t_0,T]\}$  {\em delivers} the stock GOP-denominated BN-replicable contingent claim $\tilde H_T$ at a bounded stopping time $T$ if  the equality
	\begin{equation}
	\tilde S^\delta_T=\tilde H_T
	\end{equation}
holds $Q_{S^*}$-almost surely.
\end{definition}
Combining  Definition \ref{integr'}, Definition \ref{HT}, Definition \ref{deliver}, and the SDE \eqref{stochint},   leads  directly to the following statement:
\begin{corollary}
	For a BN-replicable stock GOP-denominated contingent claim $\tilde H_T$ with representation \eqref{tildeHT}, there exists an admissible self-financing strategy $\delta_{\tilde H_T}=\{\delta_{\tilde H_T}(t)=(\delta_{\tilde H_T}^0(t),...,
	\delta_{\tilde H_T}^d(t))^\top,t\in[t_0,T]\}$ with corresponding stock GOP-denominated price process $\tilde S^{\delta_{\tilde H_T}}_t=\tilde H_t$ given by the {\em benchmark-neutral pricing formula}
	\begin{equation}
	\tilde H_t={\bf E}^{Q_{S^*}}(\tilde H_{T}|\mathcal{F}_t),
	\end{equation} which
	delivers the stock GOP-denominated contingent claim	\begin{equation}
\tilde H_T=\tilde S^\delta_T
	\end{equation}
	$Q_{S^*}$-almost surely.
\end{corollary}
 The stock GOP-denominated price 	$\tilde H_t$
	at time $t\in[t_0,T]$ yields, within the set  of admissible self-financing strategies, the minimal possible self-financing portfolio process that delivers the stock GOP-denominated contingent claim $\tilde H_T$.\\

\subsection{Hedging Strategy}
Recall that the $j$-th ˆ
stock GOP-denominated primary security account process $\tilde S^j_t$,  $j\in\{0,1,...,d\}$,  is a $(Q_{S^*}, \mathcal{\underline F})$-local martingale under the benchmark-neutral pricing measure $Q_{S^*}$. Consequently,
a stock GOP-denominated self-financing portfolio $\tilde S^\delta_t$ is  a $(Q_{S^*}, \mathcal{\underline F})$-local martingale. \\
Consider a stock GOP-denominated BN-replicable contingent claim $\tilde H_T$ with a bounded stopping time $T$ as
maturity, where  its entire randomness is driven by the $(Q_{S^*}, \mathcal{\underline F})$-local martingales $W^0,W^1,...,W^{d-1}$.
These local martingales are assumed to be $Q_{S^*}$-orthogonal to each other 
in the sense that their pairwise products form $(Q_{S^*}, \mathcal{\underline F})$-local martingales. Furthermore, each
stock GOP-denominated primary security account value $\tilde S^j_t$,  $j\in\{0,...,d\}$, is assumed to satisfy an SDE  of the form
\begin{equation}
\frac{d \tilde S^j_t}{\tilde S^j_t}=-\sum_{k=0}^{d-1}\theta^{j,k}_tdW^k_t
\end{equation}
for $t\in[t_0,\infty)$ with $\tilde S^j_{t_0}>0$. We assume that $\theta^{j,k}=\{\theta^{j,k}_t,t\in[t_0,\infty)\}$  forms for each $j\in\{0,1,...,d\}$ and $k\in\{0,1,...,d-1\}$
a predictable process such that the above stochastic differentials are well defined, see \citeN{KaratzasSh98}. Note that $W^0$ is the $(Q_{S^*}, \mathcal{\underline F})$ Brownian motion introduced in \eqref{W0}, and by \eqref{e.4''} we have  $\theta^{0,0}_t=\theta^0_t$ and $\theta^{0,k}_t=0$ for $k\in\{1,...,d-1\}$ and $t\in[t_0,\infty)$. \\
 For  $t\in[t_0,\infty)$ we denote by  $\Phi_t=[\Phi^{j,k}_t]_{j,k=0}^{d,d}$ the matrix with elements
\begin{equation}
\Phi^{j,k}_t=\theta^{j,k}_t
\end{equation}
for $j\in\{0,1,...,d\}$ and $k\in\{0,1,...,d-1\}$, and
\begin{equation}
\Phi^{j,d}_t=1
\end{equation}
for $j\in\{0,1,...,d\}$. \\Let us make the following assumption:
\begin{assumption} We assume that a BN-replicable stock GOP-denominated contingent claim $\tilde H_T$ has for its fair stock GOP-denominated price at time $t\in[t_0,T]$ under $Q_{S^*}$ a unique martingale representation  of the form
\begin{equation}\label{tildeHT''}
\tilde H_t=\tilde H_{t_0}+\sum_{k=0}^{d-1}\int_{t_0}^{t}x^k_s d{W^k_s},
\end{equation}
where $x^0,x^1,...,x^{d-1}$ are predictable and the integrals
\begin{equation}
\int_{t_0}^{t}(x^k_s)^2ds<\infty
\end{equation} are $Q_{S^*}$-almost surely finite for every $k\in\{0,1,...,d-1\}$ and $t\in[t_0,\infty)$. 
\end{assumption}  As in the proof of Proposition 7.1 in \citeN{DuPl16}, it follows:
\begin{theorem}  If the matrix $\Phi_t$ is Lebesgue-almost everywhere invertible, then the  strategy $\delta_{\tilde H_T}$ is given by the vector
\begin{equation}
\delta_{\tilde H_T}(t)=diag({\bf{\tilde S}}_t)^{-1}(\Phi_t^\top)^{-1}\xi_t
\end{equation}
with
\begin{equation}
\xi_t=(-x^0_t,-x^1_t,...,-x^{d-1}_t,\tilde H_t)^\top
\end{equation}
for all $t\in[t_0,T)$.
\end{theorem}
Here $diag({\bf{ S}})$ denotes the diagonal matrix with the elements of a vector ${\bf{ S}}$ as its diagonal. In the case when the  dynamics of the extended market are modeled using state variables that satisfy the SDEs of a Markovian system of  diffusions, one can systematically identify for a given stock GOP-denominated BN-replicable contingent claim $\tilde H_T$ the respective representation  \eqref{tildeHT}. The price  $\tilde H_t$ at the time $t$ can be obtained,  e.g., by
explicit calculation of the conditional expectation, by application of the Feynmnan-Kac formula, or via some numerical method.  The price results as a function of the state variables that satisfies a respective partial differential equation (PDE). \\  The integrands in the representation \eqref{tildeHT''} and  the predictable vector $\xi_t$ emerge  when applying the It\^{o} formula to the price function and matching the respective terms in the martingale part of the resulting SDE. The PDE operator follows by setting the drift part in the resulting SDE for the price function to zero. 
Consequently, the price function satisfies 
a  Kolmogorov-backward PDE.\\ The boundary conditions of the PDE need  to be specified such that the solution of the PDE, as a function of the evolving state variables, becomes a $(Q_{S^*}, \mathcal{\underline F})$-martingale.  When only fixing the PDE operator and the boundary conditions that are determined by the payoff structure  of the contingent claim, there may exist several  price functions that solve the PDE. All these solutions yield nonnegative $(Q_{S^*}, \mathcal{\underline F})$-local martingales. Since these processes form $(Q_{S^*}, \mathcal{\underline F})$-supermartingales that deliver the targeted contingent claim, they yield  price processes that are larger than or equal to the one that forms the {\em fair} price process, which is the $(Q_{S^*}, \mathcal{\underline F})$-martingale.   We emphasize, it is the fair price process that delivers at time $T$ the stock GOP-denominated contingent claim $\tilde H_T$ by starting from the most economical minimal possible stock GOP-denominated initial  price $\tilde H_{t_0}$.\\

\section{Stock GOP Dynamics}\setcA

\subsection{Leverage Effect and Equivalent BN Pricing Measure}
The question arises  whether  Assumption \ref{RNDMARTINGALE}, which provides the martingale property of the Radon-Nikodym derivative of the BN pricing measure, is  realistic. This means whether it is realistic for  existing stock markets to model the Radon-Nikodym derivative $\Lambda_{S^*}(t)$ of the putative BN pricing measure $Q_{S^*}$ as a true  $(P, \mathcal{\underline F})$-martingale. The answer to this question is closely related to the boundary behavior of the   volatility  $\sigma^{S^*}(t)= \frac{\lambda^*_t}{\theta_t}$, see \eqref{sigma10}, of the Radon-Nikodym derivative $\Lambda_{S^*}(t)$; see, e.g.,  \citeN{AndersenPi07}, \citeN{HulleyPl12}, and  \citeN{HulleyRu19}. \\Intuitively,  the  num\'eraire $S^*_t$, when denominated in the GOP $S^{**}_t$ of the extended market,  is a martingale when its  volatility  $\sigma^{S^*}(t)$   remains finite for finite values of  $S^*_t$, including its asymptotic  value at the boundary where it approaches zero. 
 The stock GOP $S^*_t$ can be approximated by a well-diversified stock index; see, e.g., \citeN{PlatenRe20}.   Furthermore, for a  stock index it is well-known that it exhibits the {\em leverage effect}; see \citeN{Black76b} and \citeN{AitSahalia12}. This effect captures in an idealized manner  the  fact that the volatility  of a stock index becomes asymptotically infinite when the stock index value converges to zero.  The net risk-adjusted return $\lambda^*_t$  models an  average of the expected returns of the savings account-denominated stocks minus the squared volatility of the stock GOP, see, e.g., Theorem 3.1 in \citeN{FilipovicPl09}. When the value of the stock index converges to zero, which is the extreme of a stock market crash, the  expected returns of the savings account-denominated stocks, which cause the crash, seem unlikely to converge to infinity and  likely to remain finite for economic reasons.  Therefore, it seems unlikely that the net risk-adjusted return becomes infinite when the stock GOP value converges to zero and its squared volatility converges, due to the leverage effect, to infinity. Consequently, the volatility $\sigma^{S^*}(t)=\frac{\lambda^*_t}{\theta_t}$, see \eqref{sigma10}, of the Radon-Nikodym derivative of the BN pricing measure    remains most likely finite when the stock GOP value converges to zero. Otherwise, the net risk-adjusted return would have to converge to infinity when the stock GOP would converge to zero, which seems unlikely.
   This intuition indicates that  stock GOP dynamics, which exhibit the empirically well-observed leverage effect, are likely to yield an equivalent BN pricing measure. We confirm below that this is  the case    for   the  realistic stock GOP model we will assume.
\subsection{Minimal Market Model in Activity Time}  
 To illustrate BN pricing, the current paper assumes a model where the stock GOP evolves in  some {\em activity time}  $\tau=\{\tau_t, t\in[t_0,\infty)\}$ with {\em activity}
$
a_t\in(0,\infty)
$
for $t\in[t_0,\infty)$ starting with    the  {\em   initial activity time} $\tau_{t_0}$, where
\begin{equation}\label{tau}
\tau_t=\tau_{t_0}+\int_{t_0}^{t}a_sds.
\end{equation} Intuitively, the activity reflects the trading activity, which can be expected to show  seasonal and behavioral effects. 
 The net risk-adjusted return $\lambda^*_t$  is a Lagrange multiplier, see Theorem 3.1 in \citeN{FilipovicPl09}, and does not need to be modeled under BN pricing because it becomes removed from the stock GOP dynamics under $Q_{S^*}$. Therefore, the current paper  makes the simplifying assumption, which could be substituted by  more elaborate assumptions without changing the proposed BN pricing and hedging, that the net risk-adjusted return is   proportional to the activity, which means 
\begin{equation}\label{lambda1}
\lambda^*_t=\bar\lambda^*a_t
\end{equation}  with $\bar\lambda^* >0$ for $t\in[t_0,\infty)$.  Since the  values of the net risk-adjusted return are almost impossible to estimate, the fact that they do not matter when applying BN pricing is of  practical importance and simplifies  the implementation. Another simplification for the implementation of the BN pricing and hedging of long-term contingent claims, like long-term  zero-coupon bonds, will arise from the fact that one will not have to model the random behavior of the activity time. Still, the activity time has to be observed and its average linear evolution  estimated. \\ 
The {\em minimal market model} (MMM) has been suggested by \citeN{Platen01a} as a realistic model for the long-term stock GOP dynamics.   Under the MMM, the  volatility of the stock GOP dynamics in calendar time has the form
\begin{equation}\label{theta2'''}
\theta_{t}=\sqrt{\frac{4e^{ \tau_t}a_t}{S^*_t}}
\end{equation}
for $t\in[t_0,\infty)$.  By applying this model,  the stock GOP satisfies according to \eqref{e.4} under the BN pricing measure the SDE

\BE \label{e.4''}
dS^{*}_t
=4e^{\tau_t}d\tau_t+\sqrt{S^*_t4e^{\tau_t}}d\bar W^0_{\tau_t} \EE
with $S^*_{t_0}>0$ 
 for $t\in[t_0,\infty)$. Here $\bar W^{0}_{ \tau_t}$ represents a  Brownian motion  under $Q_{S^*}$ in activity time with stochastic differential
\begin{equation}\label{barW0}d \bar W^{0}_{\tau_t}
=\sqrt{a_t}d W^0_t
\end{equation} 
for $t\in[t_0,\infty)$.  The SDE \eqref{e.4''} shows that $S^*_t$ represents a time-transformed squared Bessel process of dimension  four; see \citeN{RevuzYo99} or equation (8.7.1) in \citeN{PlatenHe06}. One notes that its volatility in activity time exhibits  a leverage effect and is, as a 3/2 volatility model, a constant elasticity of variance model, see, e.g., \citeN{Cox75}, \citeN{Platen97d}, \citeN{Heston97}, and \citeN{Lewis00}.\\
The MMM diffusion dynamics of a squared Bessel process of dimension four, see \citeN{RevuzYo99}, for the stock GOP seems to capture its \textquoteleft natural' dynamics. Critical is the diffusion coefficient in the SDE \eqref{e.4''}, which is proportional to the square root of the stock GOP value. Such a diffusion coefficient arises in the SDE for the limit of the population size of a birth-and-death process; see \citeN{Feller71}. The current paper suggests that the \textquoteleft natural' evolution of diversified wealth can be interpreted as that of a birth-and-death process where independently wealth units  give  birth to new ones  or die. This interpretation provides a basis for the understanding of the \textquoteleft natural' dynamics of the stock GOP. The trading activity, which we reflect in the derivative of the activity time, provides the \textquoteleft speed' at which the stock GOP dynamics evolve under the MMM. Furthermore, in \citeN{Platen23} it has been shown that the  stock GOP dynamics of the MMM  maximizes the  entropy of the stationary density of the normalized stock GOP, which characterizes its most likely dynamics. These facts provide  intuitive,  mathematically founded arguments for the choice of the  MMM   in activity time as the model for the stock GOP. 
\subsection{Observed Activity Time}
To observe the  activity time,
we consider the square root of the stock GOP $\sqrt{S^*_t}$  and obtain by application of the It\^{o} formula, \eqref{e.4''}, \eqref{barW0} and \eqref{W0} the SDE 
\BE \label{e.5}
d\sqrt{S^{*}_t}=\frac{3e^{\tau_t}}{2\sqrt{S^*_t}}d\tau_t+ \sqrt{e^{\tau_t}} d \bar W^{0}_{\tau_t}
\EE
for $t\in[t_0,\infty)$.
Therefore, the quadratic variation of $\sqrt{S^{*}_t}$ becomes
\begin{equation}
[\sqrt{S^{*}_.}]_{t_0,t}=\int_{\tau_{t_0}}^{\tau_t}e^{s}ds=e^{\tau_t}-e^{\tau_{t_0}}
\end{equation}
and the activity time results in the form \begin{equation}\label{tauest}
\tau_t= 	 \ln\left(
[\sqrt{S^{*}_.}]_{t_0,t}+e^{\tau_{t_0}}\right)
\end{equation}
for $t\in[t_0,\infty)$.\\ 
For illustration, we use the same data as in \citeN{PlatenRe20}, where the market capitalization-weighted total return stock index (MCI) was constructed that is displayed in Figure \ref{FigMCI1}. The MCI was in \citeN{PlatenRe20} from the stock data directly generated as a match of the daily observed US Dollar savings account-denominated MSCI-Total Return Stock Index for the developed markets. As shown in \citeN{PlatenRe20} one can interpret the MCI as a reasonable  proxy for the stock GOP. \\ For some tentative initial activity time, one can observe the trajectory of the resulting respective activity time, which turns out to evolve approximately linearly. Only at the beginning of its trajectory, one notices some deviation from its approximate linearity, which changes with the choice of the tentative value of the  initial activity time. This leads  to the assumption that the average of the activity time is  a linear function of calendar time.  Therefore, one can search by standard linear regression for the initial activity time that yields the activity time that is as close as possible to a straight line.
   The resulting  
  estimated {\em trendline} \begin{equation}\bar \tau_t=\bar \tau_{t_0}+ \bar at
\end{equation}
  is exhibited in Figure \ref{FigMCI2} together with the MCI for the  period from $t_0= $2 January 1984 until $T= $1 November 2014. We observe for the trendline 
  its  {\em slope} $ \bar a\approx 0.053$ and {\em  initial value} $\bar \tau_{t_0}\approx 2.15$.  The $R^2$-value of $0.98$ confirms that the  observed  activity time evolves  approximately linearly. Because of this finding, the current paper   employs first
     the trendline $\bar \tau_t$  as a substitute for the activity time in BN pricing and hedging and later an enhanced way of pricing and hedging. When approximating maturity dates and other times in  activity time, the two parameters $\bar \tau_{t_0}\approx-103.86$ and $ \bar a\approx 0.053$ are assumed to remain the same also in the future. 
\begin{figure}[h!]
	\centering
	\includegraphics[width=14cm, height=6cm]{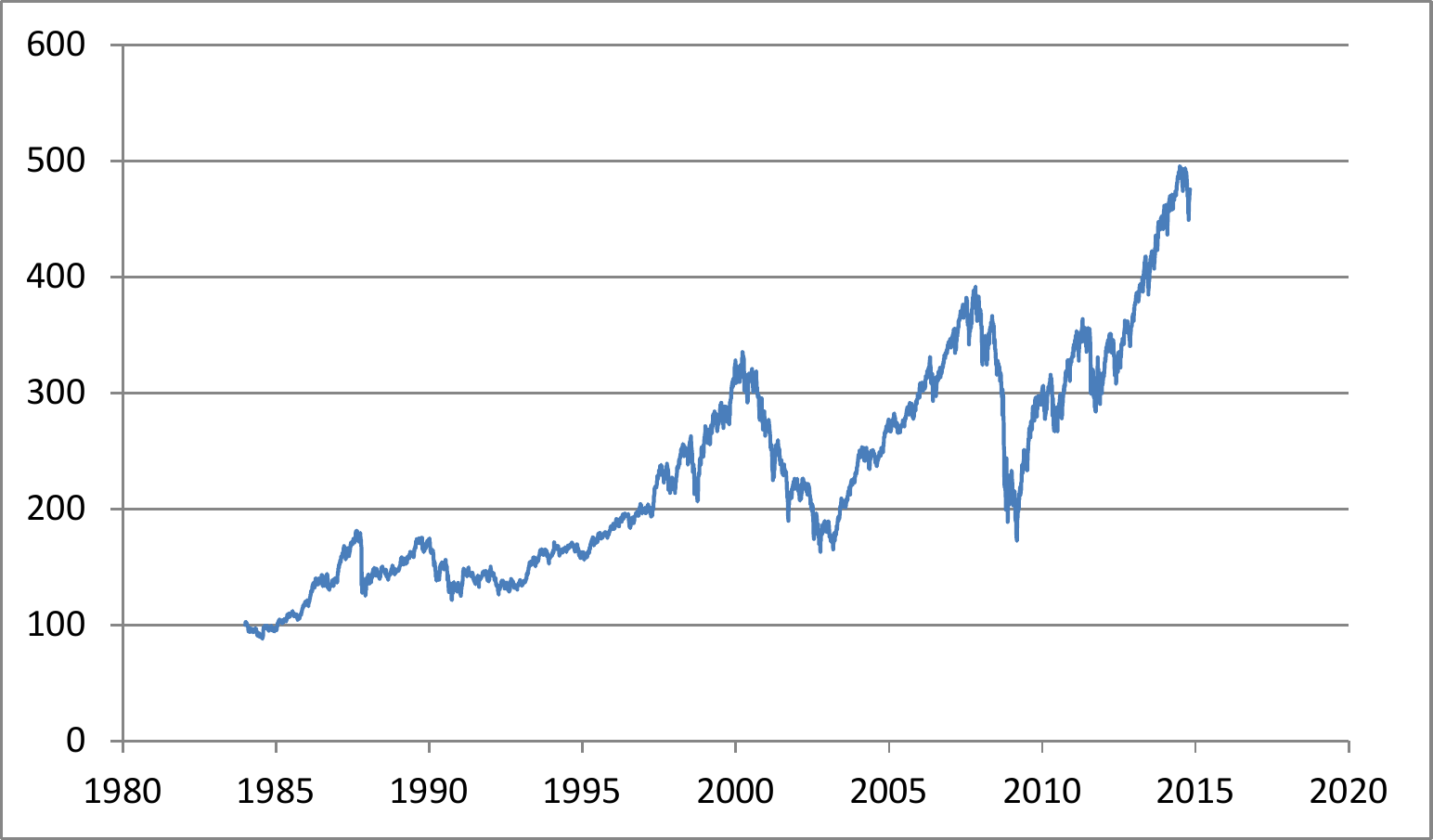}\\
	\caption{US Dollar savings account-denominated MCI.}\label{FigMCI1}
\end{figure} 
\begin{figure}[h!]
	\centering
	\includegraphics[width=14cm, height=6cm]{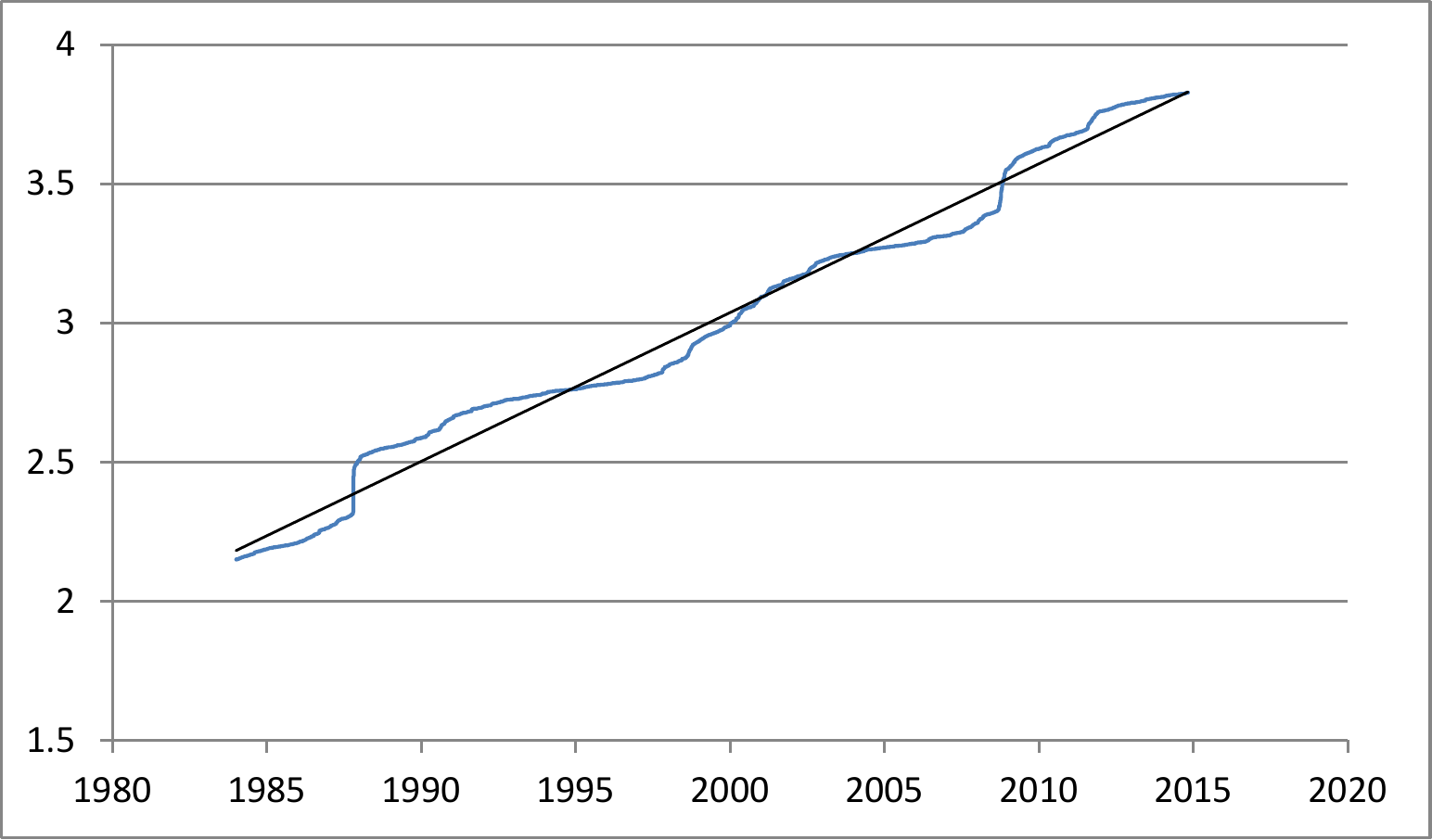}\\
	\caption{Activity time $\tau_t$  and trendline $\bar\tau_t$.}\label{FigMCI2}
\end{figure}

\subsection{BN Pricing Measure}
When using  the  trendline $\bar \tau_t$ with constant slope $\bar a\in(0,\infty)$ as model for the activity time, it follows by \eqref{sigma10} and  \eqref{theta2'''} that the volatility $\sigma^{S^*}(t)$  of the Radon-Nikodym derivative of the BN pricing measure  equals
\begin{equation}\label{sigma2'}
 \sigma^{S^*}(t)=\bar\lambda^*\sqrt{\frac{\bar aS^*_t}{4e^{\bar{\tau_t}}}}
\end{equation}
for $t\in[t_0,\infty)$.  This  allows us to prove  the following result:
\begin{theorem}\label{truemartBP}
	Under the above assumptions and the model that employs the trendline $\bar \tau_t$ as a proxy for the activity time, the Radon-Nikodym derivative of the BN pricing measure $Q_{S^*}$ is  a   $(P, \mathcal{\underline F})$-martingale and $Q_{S^*}$   an equivalent probability measure.
\end{theorem}
\begin{proof}
	Under the above assumptions it follows from \eqref{dS*T} that $S^*_t$ is the product of a  time-transformed squared Bessel process with value
	 \begin{equation}\bar S^*_t=S^*_te^{-\bar\lambda^* \bar a t}\end{equation} 
	and the deterministic exponential function of time $e^{\bar\lambda^* a t}$. Therefore, by \eqref{sigma2'}  the squared  volatility  $( \sigma^{S^*}(t))^2$
	  of the Radon-Nikodym derivative $\Lambda_{S^*}$ of the BN pricing measure $Q_{S^*}$ equals  the product
	of the deterministic exponential function of time \begin{equation}(\bar \lambda^*)^2\frac{\bar a}{4e^{\bar \tau_{t_0}+\bar a(1-\bar \lambda^*)t}}\end{equation} 
	and the   time-transformed squared Bessel process of dimension four $\bar S^*_t$. \\	When the squared volatility of a local martingale is the product of a  time-transformed squared Bessel process of dimension greater than two and a deterministic  function of time, it follows  directly that the local martingale  $\Lambda_{S^*}$ 
	is  a  $(P, \mathcal{\underline F})$-martingale by using the argument from   the proof of Proposition 2.5 in \citeN{AndersenPi07}, which is given in their Lemma 2.3. According to the Change of Num\'eraire Theorem in \citeN{GemanElRo95}, it follows that  $Q_{S^*}$ is  an equivalent probability measure, which proves the statement of Theorem \ref{truemartBP}.
\end{proof}

\subsection{Putative Risk-Neutral Pricing Measure}
Risk-neutral pricing employs the savings account $S^0_t=1$ as num\'eraire and the classical finance theory postulates that the putative risk-neutral pricing measure $Q_{S^0}$ is an equivalent probability measure, see \citeN{Cochrane01} and \citeN{DelbaenSc98}. The following result shows that  $Q_{S^0}$ is not an equivalent probability measure   under the above-described model:

\begin{theorem}\label{locmartRNP}
	Under the above assumptions, the Radon-Nikodym derivative $\Lambda_{S^0}$  of the putative risk-neutral pricing measure $Q_{S^0}$ is  a strict  $(P, \mathcal{\underline F})$-supermartingale and   $Q_{S^0}$ is not an equivalent probability measure.
\end{theorem}
\begin{proof}
	The Radon-Nikodym derivative of the putative risk-neutral pricing measure $Q_{S^0}$ equals
	\begin{equation}\Lambda_{S^0}(t)=\frac{dQ_{S^0}}{dP}|_{\mathcal{F}_t}=\frac{S^0_t}{S^{**}_t}=\frac{1}{S^{**}_t}\end{equation}
	for $t\in[t_0,\infty)$. 
It follows by Theorem \ref{truemartBP} that
	\begin{equation}
	{\bf E}^{P}(\Lambda_{S^0}(s)|\mathcal{F}_t)={\bf E}^{P}(\frac{1}{S^{**}_s}|\mathcal{F}_t)=\frac{S^*_t}{S^{**}_t}{\bf E}^{Q_{S^*}}(\frac{1}{S^{*}_s}|\mathcal{F}_t)
	\end{equation}
	for $t_0\leq t<s<\infty$. Under $Q_{S^*}$  the savings account-denominated stock GOP $S^*_{t}$ is a time-transformed squared Bessel process of dimension four satisfying the SDE \eqref{e.4}. Its inverse $\frac{1}{S^*_t}$ is   a strict $(Q_{S^*}, \mathcal{\underline F})$-supermartingale; see \citeN{RevuzYo99}. Therefore, we have 
	\begin{equation}
	{\bf E}^{P}(\Lambda_{S^0}(s)|\mathcal{F}_t)<\frac{S^*_t}{S^{**}_tS^*_t}=\frac{1}{S^{**}_t}=\Lambda_{S^0}(t)  
	\end{equation}
	for $t_0\leq t\leq s<\infty$,	which shows that $\Lambda_{S^0}$ is a strict $(P, \mathcal{\underline F})$-supermartingale. By application of the Change of Num\'eraire Theorem in \citeN{GemanElRo95}, the putative risk-neutral measure is not an equivalent probability measure, which proves Theorem \ref{locmartRNP}.
\end{proof}\\

\noindent Recall that the leverage effect provided the intuition that the Radon-Nikodym derivative of the BN-pricing measure is a true martingale. Therefore,  Theorem \ref{locmartRNP}   provides the intuition that the Radon-Nikodym derivative of the putative risk-neutral pricing measure is  unlikely  a martingale. 
 This intuition seems  to come close to reality, as we will see below.

\section{Pricing and Hedging of a Zero-Coupon Bond}\setcA
\subsection{BN Pricing of a Zero-Coupon Bond}
 Under the above model, which employs the trendline of the activity time as model for the activity time, the transition probability density for the stock GOP $S^*$ under $Q_{S^*}$ has, by following \citeN{RevuzYo99} or equation (8.7.9) in \citeN{PlatenHe06}, the form 
 \begin{equation}
 p(\bar \tau_t,S^*_0;\bar \tau_s,S^*_s)=\frac{1}{2(e^{\bar \tau_s}-e^{\bar \tau_t})}\left(\frac{S^*_s}{S^*_t}\right)^{\frac{1}{2}}\exp\left\{-\frac{S^*_t+S^*_s}{2(s^{\bar \tau_s}-e^{\bar \tau_t})}\right\}I_1\left(\frac{\sqrt{S^*_tS^*_s}}{e^{\bar \tau_s}-e^{\bar \tau_t}}\right)
 \end{equation} for $t_0\leq t\leq s<\infty$, where $I_1(.)$ denotes the modified Bessel function of the first kind with index $1$; see, e.g., \citeN{AbramowitzSt72}. Consequently, we know the transition probability density of the stock GOP  under the BN pricing measure, where its key characteristic, the trendline of the  activity time, is observable.\\
To illustrate BN pricing and hedging, we consider a  zero-coupon bond with fixed maturity $T\in(t_0,\infty)$ and contingent claim $H_T=1$.
The respective fair zero-coupon bond \begin{equation}P(t,T)=\tilde P(t,T)S^*_t\end{equation}  pays at maturity one unit of the savings account $H_T=S^0_T=1$.  Its value,  obtained via the BN pricing formula \eqref{BPF}, is given in the denomination of the savings account by the explicit formula \begin{equation}\label{bond}
P(t,T)=S^*_t{\bf{E}}^{Q_{S^*}}(\frac{1}{S^*_T}|\mathcal{F}_t)=1-\exp\left\{-\frac{S^*_t}{2(e^{\bar \tau_{T}}-e^{\bar \tau_{t}})}\right\}
\end{equation}
for $t\in[t_0,T)$; see \citeN{Platen02g} or equation (13.3.5) in \citeN{PlatenHe06}. \begin{figure}[h!]
	\centering
	\includegraphics[width=14cm, height=6cm]{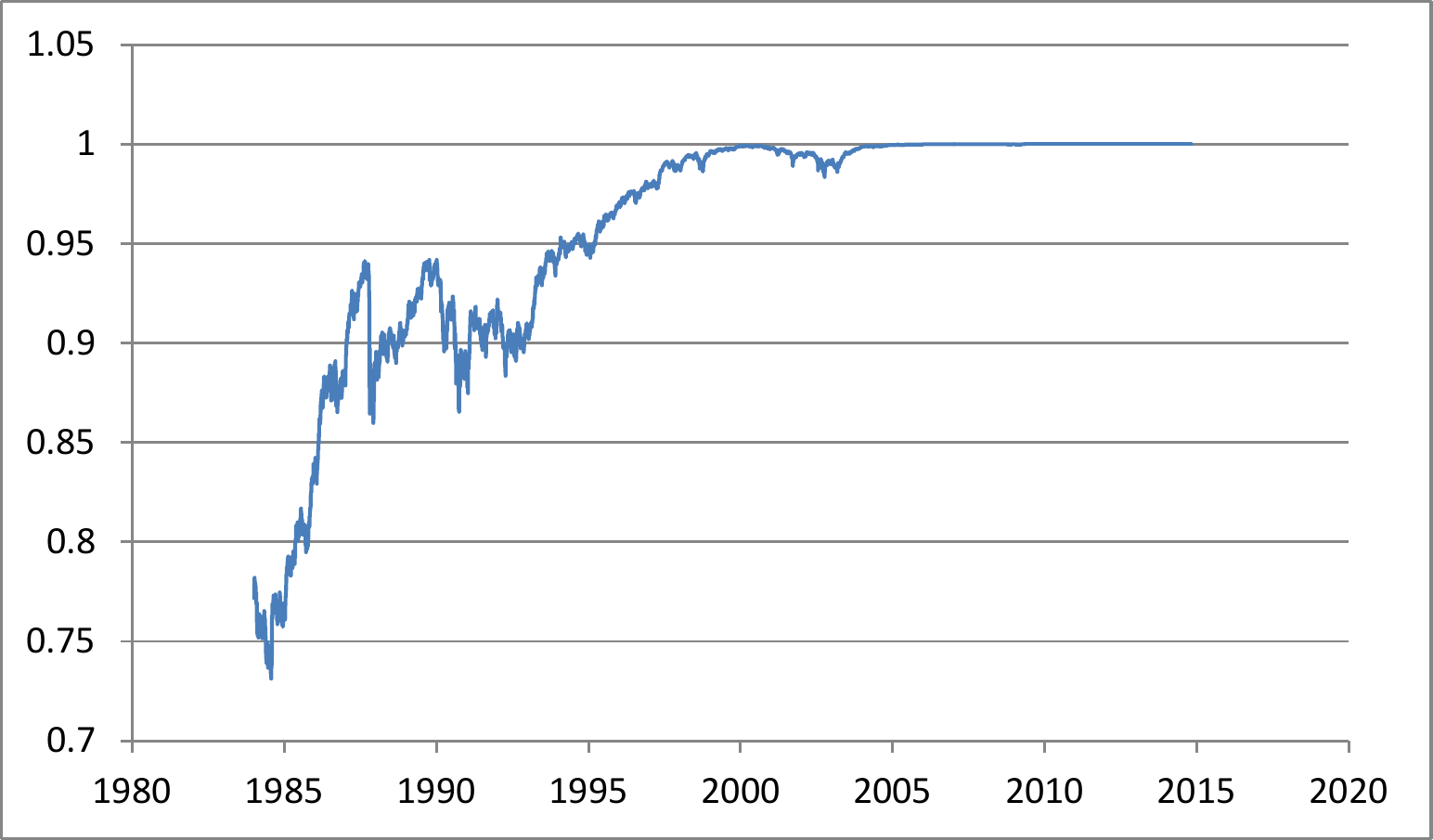}\\
	\caption{Zero-coupon bond $P(t,T)$ with MCI as the benchmark.}\label{FigMCI3}
\end{figure}
We display in Figure \ref{FigMCI3} the zero-coupon bond price with maturity at the end of our observation period that emerges when using the   MCI as a proxy for the stock GOP.  One notes that there exist at least two self-financing portfolios that hedge the payment of one unit of the savings account at the maturity $T$. The classical hedge portfolio would simply purchase one unit of the savings account at the initial time and hold it until maturity. The one that the proposed BN pricing with the MCI as benchmark suggests is less expensive, as  Figure \ref{FigMCI3} shows. It requests only about three-quarters of the  risk-neutral price.
\subsection{BN Hedging of a Zero-Coupon Bond}
The payoff of the zero-coupon bond can be replicated through hedging: We have for the zero-coupon bond $\tilde P(t,T)$, when denominated in the  stock GOP, the SDE

\begin{equation}
d\tilde P(t,T)=\delta^{ S^0}_td\tilde S^0_t
\end{equation}
with the hedge ratio
\begin{equation}\label{delta0}
\delta^{ S^0}_t=\frac{\partial \tilde P(t,T)}{\partial \tilde S^0_t}=1-\exp\left\{-\frac{S^*_t}{2(e^{\bar \tau_{T}}-e^{\bar \tau_{t}})}\right\}\left(1+\frac{S^*_t}{2(e^{\bar \tau_{T}}-e^{\bar \tau_{t}})}\right)
\end{equation}
for the investment in the savings account  $\tilde S^0_t=\frac{1}{S^*_t}$, when denominated in the stock GOP, at time $t\in[0,T)$.
 Figure \ref{FigMCI4} shows the fraction 
\begin{equation}
\pi^{S^*}_t=1-\pi^{S^0}_t=1-\frac{\delta^{S^0}_t \tilde S^0_t}{\tilde P(t,T)}=\left(\frac{1}{P(t,T)}-1\right)\frac{S^*_t}{2(e^{\bar \tau_T}-e^{\bar \tau_t})}
\end{equation}
of the value of the hedge portfolio invested in the stock GOP.
\begin{figure}[h!]
	\centering
	\includegraphics[width=14cm, height=6cm]{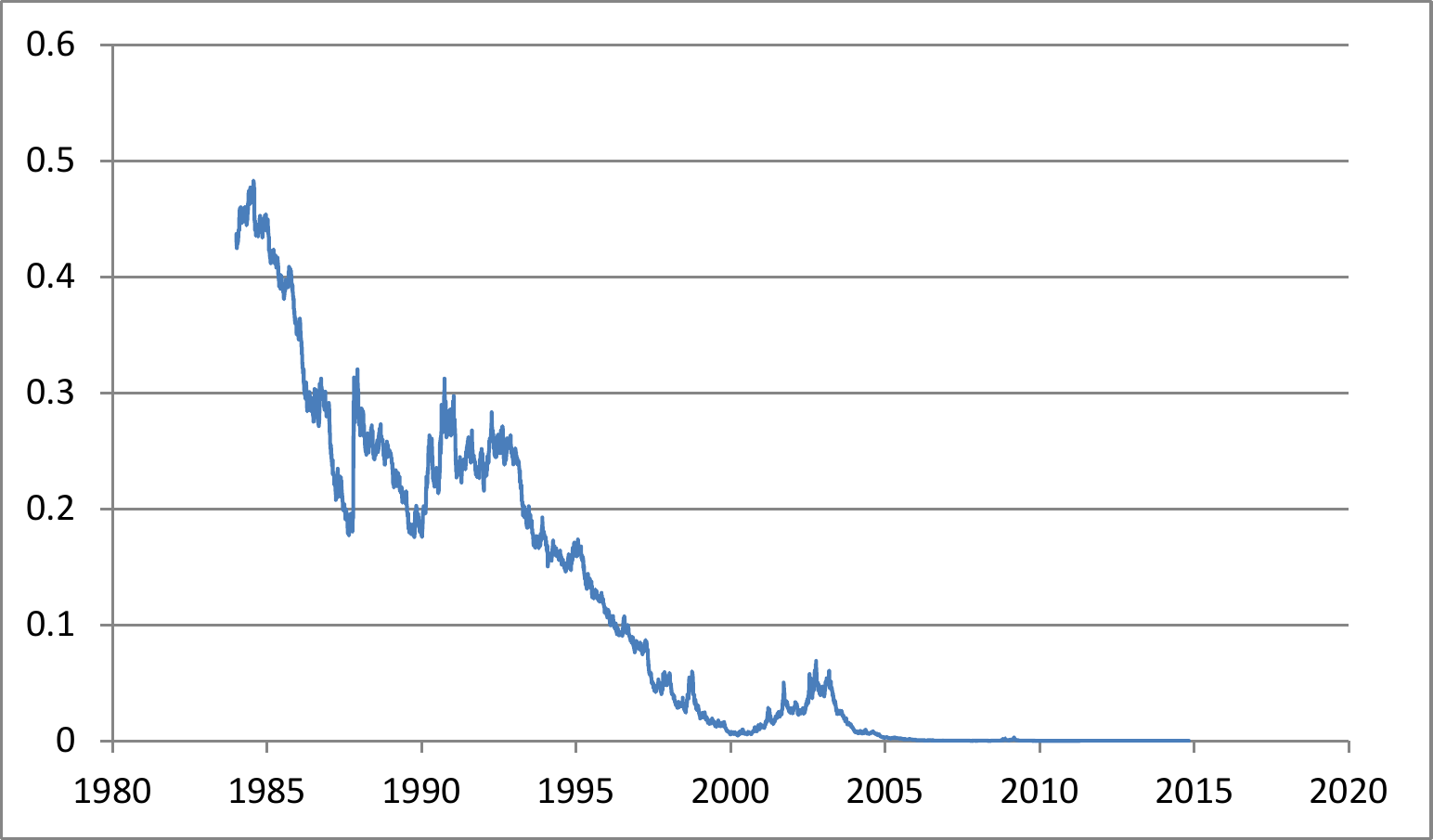}\\
	\caption{Fraction $\pi^{S^*}_t$ invested in the stock GOP.}\label{FigMCI4}
\end{figure}
 One notes  that when the time to maturity is long, the fraction invested in the stock GOP is relatively high. This fraction declines over time according to the prescribed strategy and becomes finally zero at maturity. One could interpret its trajectory as a rigorous description of the glide path for the common financial planning strategy  when it targets at maturity one unit of the savings account and one invests when young in stocks and closer to retirement more and more in the savings account.\\  Note in equation \eqref{delta0} that the strategy sells units of the savings account and buys  units of the stock GOP when the stock GOP value declines.  Since the normalized stock index is under the MMM mean-reverting, this   strategy is rational. However, the   reaction of many investors who invest for retirement appears  to be different in situations when the stock market crashes. The above rational strategy, when widely implemented, e.g., by pension funds and life insurance companies,  could potentially help to stabilize a stock market in times of a major market drawdown. \\
 
\noindent One can only trade at discrete times,  which  creates hedge errors. According to the above-described hedging strategy, the hedge portfolio process $V=\{V_t, t\in[t_0,T]\}$  reallocates daily the holdings in the  savings account and the stock GOP in a self-financing manner. The difference between the BN price  and the hedge portfolio value is shown in Figure \ref{FigMCI5}, which we call the {\em profit and loss}  \begin{equation}
C_t=V_t-P(t,T),
\end{equation}
of the hedge portfolio $V_t$ formed when replicating the zero-coupon bond which was initiated on $t_0=$1 January 1984, reallocated daily in a self-financing manner, and  matured at $T=$ 1 November 2014. \begin{figure}[h!]
	\centering
	\includegraphics[width=14cm, height=6cm]{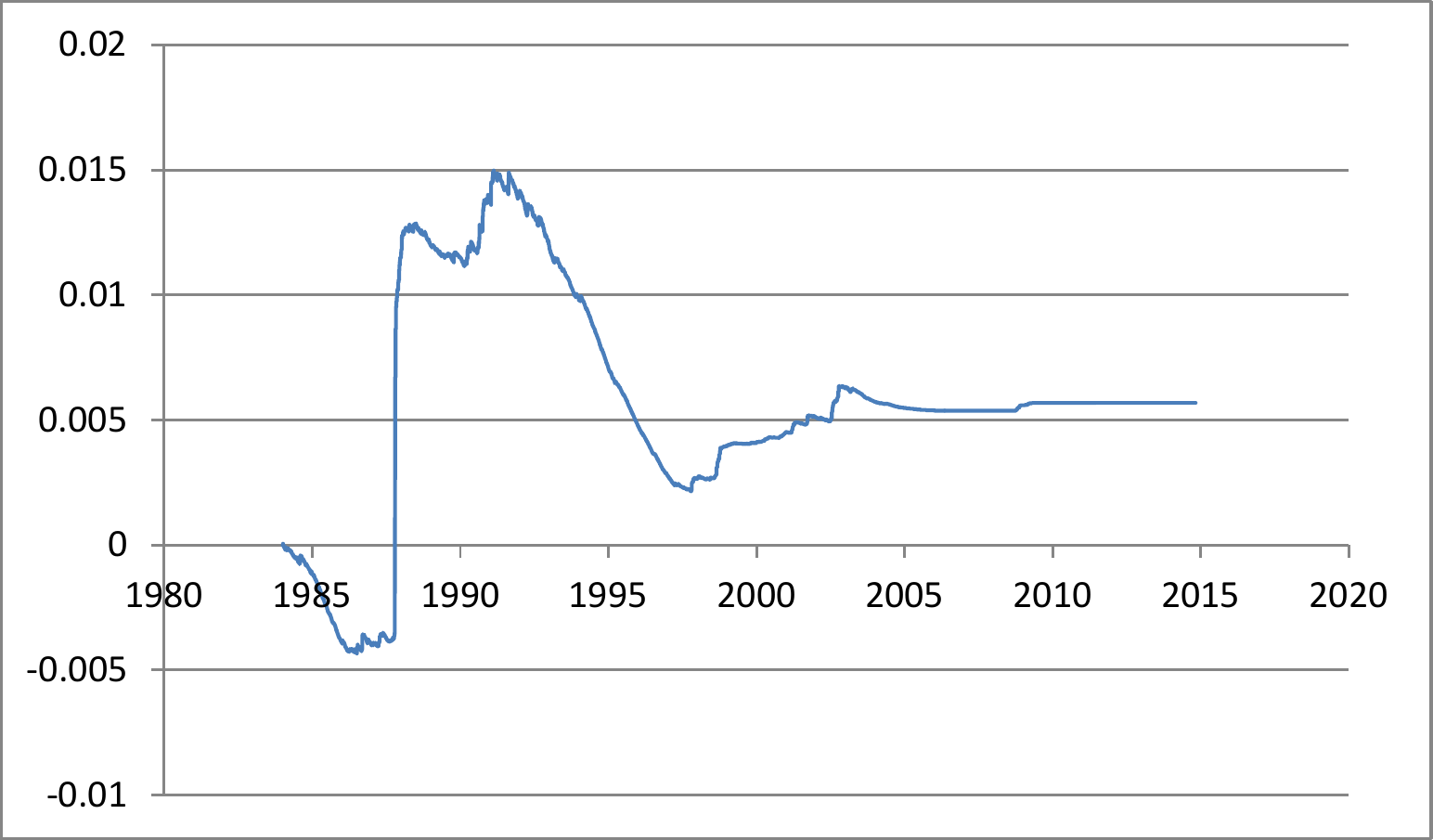}\\
	\caption{Profit and loss $C_t$ for the hedge of a zero-coupon bond.}\label{FigMCI5}
\end{figure}
The absolute value of the profit and loss turns out to be rather  small, which indicates that the  dynamics of the MMM model  well the volatility dynamics of the MCI.  The absolute value of the profit and loss remains in Figure \ref{FigMCI5} always below $1.5$\% of one unit of the savings account value, which is a small   hedge error for a hedging period of 30 years.

\subsection{Enhanced Pricing and Hedging of  Zero-Coupon Bond}
The current paper proposes  a new method that allows to reduce significantly  the hedge error when pricing and hedging a zero-coupon bond.  As can be seen in Figure \ref{FigMCI2}, the observed activity time $\tau_t$ is  different from its trendline $\bar\tau_t$. 
We introduce for the case $\tau_{t_0}<\bar \tau_T$ the stopping time 
\begin{equation}
\rho=\sup\{t\in(t_0,T]:\tau_t<\bar \tau_T\}
\end{equation} as the supremum of all times $t$ where the activity time $\tau_t$ is still smaller than the value $\bar \tau_T$ of the trendline   at maturity.  During  hedging,  one can exploit  the  information that becomes available  through the evolving observed activity time $\tau_t$ and one knows when one has reached the stopping time $\rho$. For $t\in[t_0,\rho)$, the current paper proposes an approximate  formula for the {\em enhanced zero-coupon bond price} $\bar P(t,T)$   in the form
\begin{equation}\label{bond'}
\bar P(t,T)=1-\exp\left\{-\frac{S^*_t}{2(e^{\bar \tau_{T}}-e^{ \tau_{t}})}\right\},
\end{equation}
which yields  the  {\em enhanced fraction} 
\begin{equation}\label{delta'}
\bar \pi^{S^*}_t=\left(\frac{1}{\bar P(t,T)}-1\right)\frac{S^*_t}{2(e^{\bar \tau_T}-e^{ \tau_t})}
\end{equation}
to be invested in the stock GOP. One  stops the hedge at the time $\rho$, where one exchanges all the wealth in the hedge portfolio into units of the savings account. The latter value can be expected to be    close to $1.0$ when the time step size of the hedge was sufficiently small. We denote the resulting hedge portfolio process by $\bar V=\{\bar V_t, t\in[0,T]\}$. In comparison to the previous formulas \eqref{bond} and \eqref{delta0},  in the above two formulas, \eqref{bond'} and \eqref{delta'}, we  substituted  the trendline of the activity time $\bar \tau_t$ by the  observed current activity time $\tau_t$, as long as we have $t<\rho$. The profit and loss \begin{equation}\bar C_t=\bar V_t-\bar P(t,T)\end{equation} of the enhanced hedge  portfolio $\bar V_t$ is shown in Figure \ref{FigMCI6}.  \begin{figure}[h!]
 	\centering
 	\includegraphics[width=14cm, height=6cm]{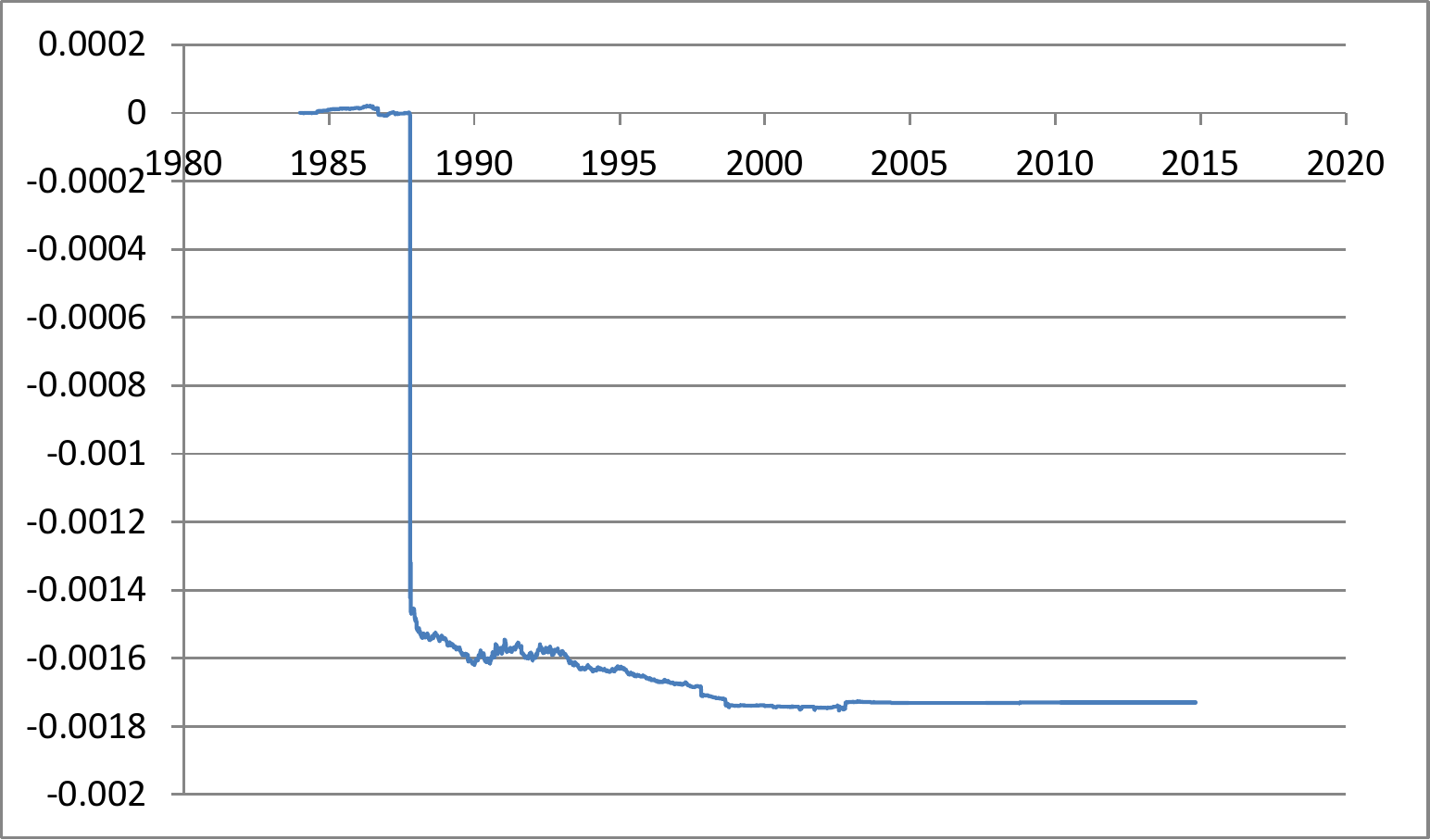}\\
 	\caption{Profit and loss $\bar C_t$ for the enhanced hedge with MCI.}\label{FigMCI6}
 \end{figure} The maximum of its absolute value remains smaller than $0.0018$ of one unit of the savings account for the 30-year daily hedge. The enhanced  hedge provides a remarkably accurate replication of the payout of the zero-coupon bond. Only during the 1987 stock market crash, where days of data  were missing, there is a sudden  increase in the absolute value of the hedge error observable. The otherwise almost perfect hedge indicates that the MMM in activity time models extremely well the   volatility dynamics of the MCI.
 \\
The risk-neutral price   before maturity for the above zero-coupon bond, which pays at maturity one unit of the savings account,  equals always one unit of the savings account.  Formula \eqref{bond} shows that when there is some strictly positive time to maturity,  its  BN price is lower than the risk-neutral price. The respective risk-neutral hedging portfolio is self-financing and delivers the targeted payoff of one unit of the savings account at maturity.  However, due to the strict supermartingale property of the Radon-Nikodym derivative process $\Lambda_{S^0}$ of the putative risk-neutral measure $Q_{S^0}$, the respective risk-neutral price   is  higher than the BN price. \\
     There exists no economic reason to produce the zero-coupon bond  payoff more expensively than necessary by following the popular  risk-neutral pricing rule. The BN price and hedge offer a  more economical way of replicating the targeted long-term payoff, as illustrated in  Figure \ref{FigMCI3}.\\ It will be separately demonstrated by employing faster-growing stock portfolios as benchmark and considering zero-coupon bonds with longer terms to maturity, the price of a zero coupon can become a rather small fraction of one unit of the savings account and the zero-coupon bonds can still be hedged as accurately  as shown above.

  \section*{Conclusion}
The  paper proposes the new method of benchmark-neutral  pricing and hedging, which employs as  num\'eraire the growth optimal portfolio  of the investment universe formed by  the stocks that can be approximated by a well-diversified stock portfolio. For a realistic   model and a long-term zero-coupon bond, it is demonstrated that the proposed benchmark-neutral price can be significantly lower  than the risk-neutral price and the payoff can be accurately hedged over several decades. By applying and extending the proposed benchmark-neutral methodology, it should be   possible to develop  accurate quantitative methods for    a wide range of long-term contracts like the risk-neutral quantitative methods that serve  short-term derivative pricing and hedging.  These new quantitative methods are urgently needed to reduce the costs of pension and insurance contracts.   The next steps toward the wider application  of benchmark-neutral pricing and hedging could target a wide range of  long-term contingent claims, including nonreplicable   claims. One could also revisit the pricing and hedging of short-term derivatives, currently performed under the putative risk-neutral pricing measure, which may provide interesting insights.

 \section*{Acknowledgements}
 The author would like to express his gratitude for
 receiving valuable suggestions on the  paper by  Kevin Fergusson, Len Garces, and Stefan Tappe.
 \section*{Data}
 The EXCEL file that generates the figures of the paper is available at Harvard Dataverse https://doi.org/10.7910/DVN/R9I0AB.

\bibliographystyle{ chicago}
\bibliography{ my}

\newpage

 \end{document}